%% file: main.tex
\documentclass[%
 reprint,
 amsmath,amssymb,
 aps,
]{revtex4-2}

\usepackage{enumitem}
\usepackage{graphicx}
\usepackage{subcaption}
\usepackage{amsmath}
\usepackage{amsthm}
\usepackage{physics}
\usepackage[normalem]{ulem}
\usepackage{mathtools}
\usepackage{tikz}
\usepackage{dcolumn}
\usepackage{bm}
\usepackage[colorlinks=true,
            linkcolor=blue,
            citecolor=blue,
            urlcolor=blue]{hyperref}
\usepackage{ragged2e}
\newtheorem{theorem}{Theorem}
\allowdisplaybreaks
\usetikzlibrary{quantikz2}

\graphicspath{...\graphs}

\begin{document} 

\title{Trade-off between predictability and quantum coherence for multi-path interferometry and its operational interpretation}

\author{Ezra Acalapati}
 \email{acalapati.ezra@phys.ens.fr}
\affiliation{%
 Laboratoire de Physique de l’\'{E}cole Normale Sup\'{e}rieure, ENS, Universit\'{e} PSL, CNRS, Sorbonne Universit\'{e}, Universit\'{e} de Paris, F-75005 Paris, France
}


\begin{abstract}

The complementarity principle is a cornerstone of quantum mechanics. In this work, we investigate a complementarity relation between the predictability and quantum coherence, respectively, representing the particle-like and wave-like behaviour in wave--particle duality, in a multi-path setup. We introduce a basis-dependent predictability defined by the Bures distance between the state dephased in the chosen basis and the maximally mixed state. The predictability depends only on the observed basis statistics and admits a closed form in terms of the Bhattacharyya overlap. We derive a trade-off relation between the  predictability and a coherence measure defined based on the nonclassicality of the Kirkwood--Dirac quasiprobability. For pure states, the trade-off relation is an exact equality.
Remarkably, this wave--particle duality relation endows the coherence measure with an operational interpretation as the classically irreducible part of measurement randomness, yielding a tight worst-case bound on the guessing probability in source-independent QRNG.

\end{abstract}
	
	\maketitle
	\flushbottom

\section{Introduction}
\label{sec: IntWPD}

Bohr's complementarity principle serves as a fundamental tenet in quantum theory~\cite{Bohr:1928}. It underlies the wave--particle duality in the double-slit experiments ~\cite{Wootters:1979, Greenberger:1988aih, Englert:1996zz}. The principle asserts that it is impossible to observe the particle-like and wave-like behaviour of the system simultaneously, as expressed quantitatively, e.g. in Englert-Greenberger-Yasin trade-off relation between path-distinguishability and interference visibility  ~\cite{Greenberger:1988aih, Englert:1996zz}. This duality has been experimentally observed in diverse platforms including trapped atoms~\cite{Wang:2016mbg}, photonic interferometers~\cite{PhysRevLett.119.040401, Yuan:2018qwy, Gao:2018pmj, Yoon:2021dhw, Chen:2022dqt}, a superconducting circuit~\cite{Huang:2022fdp} and a nanophotonic quantum chip~\cite{Ding:2025}. Beyond its fundamental theoretical significance, complementarity has recently emerged as a core element in quantum cryptographic security proofs and protocols ~\cite{Koashi:2005nki, Mizutani:2017rzz, Gao:2021gzc, Spegel-Lexne:2024mmz, Raj:2025ypf}. Furthermore, related complementarity-based frameworks have found critical applications in quantum sensing, particularly within induced-coherence and undetected-photon imaging protocols ~\cite{Kalashnikov:2015bba, Paterova:2017vxw, Kutas:2019kng, Kviatkovsky:2020gku, Gemmell:2022qeo}.

It is natural to expect that the wave--particle duality extends to multi-path interferometry ~\cite{Sorkin:1994dt, Sinha:2010mwa, Qureshi:2017ytw, Roy:2019oyy, Qureshi:2019xrb, Giordani2023IntegratedPhotonics}. In this setup, interference visibility and path-distinguishability no longer have analytical forms. Many attempts have been made to use different quantities to capture the wave-like and particle-like behaviour. An intuitive approach is to use quantum coherence \cite{Streltsov:2016iow} to quantify wave-like behaviour ~\cite{Bera:2015wbe}. Other approaches derive complementarity relations via entropic bounds~\cite{Coles:2016yrt, Bagan:2020}, optimal which-path discrimination alongside mutual information~\cite{Melo:2026nix}, incoherent operations~\cite{Liu:2024oxh} and metric/geometry-based approaches~\cite{Yang:2024amw}. Relatedly, several works extend duality to triality by including mixedness as an additional term~\cite{Jakob:2007zz, Qian:2018, Qian:2020, Basso:2021hjj, Maziero:2022dkw, Ding:2025}.

A recent approach to study wave--particle duality is to frame it as an information-theoretic guessing game ~\cite{Coles:2016yrt, Bagan:2017dtn}. In this scenario, an agent attempts to extract the certainty in getting particle-like (which-path) information with the wave--particle duality relation imposing fundamental limits on the performance. This operational interpretation of the wave--particle duality in terms of a guessing game is particularly relevant for source-independent QRNG, where the source is untrusted while the measurement device is assumed to be trusted ~\cite{Cao:2016vsg, Ma:2019aol}. In this scenario, the duality relation can be turned into a worst-case randomness bound.

In this work, we shall formulate a multi-path wave--particle duality relation which admits direct operational interpretation in such a guessing game. First, given a reference basis, we introduce predictability as a measure of particle-like behaviour in the multi-slit interferometry. It is defined as the Bures distance between a dephased state with respect to the basis and the maximally mixed state. We show that the predictability is exactly complementary to a coherence measure for quantifying the wave-like behaviour that is defined in terms of the nonclassical values of the Kirkwood--Dirac quasiprobability \cite{Budiyono:2023riv,Budiyono:2023vcn}. The duality relation suggests an operational interpretation of the coherence measure as the part of the measurement uncertainty that cannot be removed by revealing any classical preparation label. This further leads naturally to bounds on guessing probability in a protocol for source-independent QRNG with classical side information.

\section{Predictability and Coherence}
\label{sec: PreCoh}

\subsection{Predictability from Bures distance between dephased state and maximally mixed state}
\label{subsec: BurDisPreMea}

Predictability is a classical \textit{a priori} notion. It quantifies how sharply peaked the outcome statistics are for a chosen measurement basis~\cite{Roy:2024lmo, Basso:2021eax, Maziero:2022dkw}. Consider a $d$-path experiment. Suppose that a particle traverses path $a$ with probability $p_a$, $a=1,\dots d$. If the distribution is uniform ($p_a=1/d$), no path is preferred, and the path is maximally uncertain and hence predictability is minimal. At the opposite extreme, if one path occurs with certainty ($p_{a^\ast}=1$), the path is fully determined, and predictability is maximal. A predictability quantifier is required to measure the gap between maximal and minimal predictability.

A natural way to quantify this gap between `certain' and `uncertain' outcome statistics is via a distance between probability distributions. The Bhattacharyya overlap (equivalently, the Bhattacharyya coefficient) captures how distinguishable two distributions are~\cite{Bhattacharyya:1943, Fuchs:1997ss} as it is directly tied to statistical distinguishability. Its quantum analogue is fidelity, which leads to the concept of the Bures distance, which lets us lift this statistical notion of distinguishability directly to quantum states, in the spirit of earlier geometric approaches based on trace- and Hellinger-type distances~\cite{Roga:2016, Markham:2018, Puchala:2016}.

In order to promote the classical idea to quantum, the predictability needs to be basis-dependent. For a quantum state $\rho$ and projective measurement basis $\{\Pi_a\}$, it should depend only on the observed measurement statistics $p_a=\Tr(\rho\Pi_a)$. We therefore evaluate it on the dephased state
$\Delta_a(\rho)=\sum_a \Pi_a \rho \Pi_a$,
which erases phases while leaving $\{p_a\}$ unchanged~\cite{Durr:2001, Englert:2008}. We then define predictability as the Bures distance between $\Delta_a(\rho)$ and the maximally mixed state for a $d$-dimensional system,
\begin{equation}
    \label{eq: BurDisPreMea}
    P(\rho, \{\Pi_a\}) \coloneqq D_B^2 (\Delta_a(\rho) \| I/d),
\end{equation}
where $P(\rho, \{\Pi_a\})$ quantifies the predictability of an arbitrary state $\rho$ with respect to the projective measurement ${\Pi_a}$ and the (squared) Bures distance $D_B^2(\rho \| \sigma) = 2 \left( 1- \sqrt{F(\rho,\sigma)}\right)$, which includes the Uhlmann fidelity $F(\rho,\sigma) = \left( \Tr \sqrt{\sqrt{\rho}~\sigma\sqrt{\rho}} \right)^2$.
Because $\Delta_a(\rho)$ is diagonal in the chosen basis, Eq. \eqref{eq: BurDisPreMea} acquires a simple form in terms of the classical probability distribution $\{ p_a \}$:
\begin{equation}
    \label{eq: BurDisPreMeaNum}
    D_B^2 (\Delta_a(\rho) \| I/d) = 2 \left( 1 - \frac{1}{\sqrt{d}} \sum_a \sqrt{p_a} \right).
\end{equation}
Meanwhile, the maximal distance happens when only a single outcome is possible. This occurs when $\rho$ is an eigenstate of ${\Pi_a}$, giving $2 \left( 1 - \frac{1}{\sqrt{d}} \right)$. 
Finding the ratio between the  predictability and its maximal value yields the expression for a normalised predictability,

\begin{equation}
    \label{eq: BurDisPreMeaAlt2}
    \widetilde P(\rho, \{\Pi_a\}) \coloneqq \frac{P(\rho, \{\Pi_a\})}{\max_\rho P(\rho, \{\Pi_a\})} = \frac{\sqrt{d} - \sum_a \sqrt{p_a}}{\sqrt{d} - 1}.
\end{equation}

Following the framework introduced in ~\cite{Durr:2001, Englert:2008, Basso:2021eax} for a \textit{resource theory of predictability} (RTP), the  predictability needs to satisfy several properties. This classical monotone has the natural properties of a certainty measure: (1) maximum when certain (i.e. $p_{a^*} = 1$), (2) minimum when maximally uncertain (i.e. $p_a = 1/d$), (3) invariant under basis permutation since it is expressed in terms of the dephased state $\Delta_a(\rho)$ and (4) monotonicity under mixing (see App. \ref{app: Pro56}).

\subsection{General Geometric Coherence from Pure-State Coherence}
\label{subsec: PreNCLKD}

Coherence is a basis–dependent manifestation of quantum superposition. A coherence quantifier should therefore vanish on $\mathcal I$ and increase as the off–diagonal (superposition) elements with respect to ${\Pi_a}$ grow. This basis-relative viewpoint is standard in the modern resource theory of coherence~\cite{Baumgratz:2013ecx, Streltsov:2016iow, WinterYang2016OperationalCoherence} and will be our backdrop for turning the geometric predictability of section \ref{subsec: BurDisPreMea} into a genuine coherence monotone~\cite{Streltsov:2016iow}.

One of the forms of coherence is the \textit{geometric coherence},
which is defined using \textit{fidelity} $F$. For illustration, an example of geometric coherence is $C_g(\rho) = 1 - \max_{\sigma \in \mathcal{I}} F(\rho, \sigma)$, where $\mathcal{I}$ is the set of incoherent states~\cite{Streltsov:2016iow, Indrajith:2021ybu, Lu_2025}. The expression finds the optimal incoherent state such that it gives the maximal value. Another common expression is by encapsulating fidelity with a square root~\cite{Streltsov:2016iow, Liu:2017eof}.
These expressions use the negative proportion of fidelity, and the reason behind it is that fidelity is concave, which gives the convexity requirement for the coherence. In addition, the geometric coherences above obey the monotonicity property of coherence.

The geometric  predictability in section \ref{subsec: BurDisPreMea} is related to the Kirkwood–-Dirac (KD) coherence, which follows the recent quasiprobability-based approach to coherence~\cite{Budiyono:2023riv, Budiyono:2023vcn},
in its pure-state form for a $d$-dimensional pure state $\ket{\psi}$ and projective measurement basis $\{ \Pi_a \}$,
\begin{equation}
    \label{eq: KDCohPur}
    C^{NCl}_{KD}(\ket{\psi}, \{\Pi_a\}) = -1 + \sum_{a}\sqrt{|\bra{a}\ket{\psi}|^2}.
\end{equation}
The coherence is expressed as the overlap between a basis state and the target state, and it is no other than the definition of fidelity for pure states.

We offer a new geometric coherence that is positively proportional to fidelity, with its convexity guaranteed by the \textit{convex–roof construction}~\cite{Uhlmann:2010, Gour:2024znp}. This definition admits interpretations that are not explicit in the standard formulation of geometric coherences.
In this method, we decompose the target state into an arbitrary number of pure states $\rho = \sum_i p_i \ket{\psi_i}\bra{\psi_i}$, and the resulting coherence is the minimum of the average of the coherence on each pure state weighted by $p_i$. This gives the following expression:
\begin{equation}
    \label{eq: ConRooCoh}
    C(\rho, \{\Pi_a\}) = \inf_{\{p_i, \ket{\psi_i}\}} \sum_{i} p_i C(\ket{\psi_i}, \{\Pi_a\}).
\end{equation}
We can also define the normalised coherence
\begin{equation}
    \label{eq: CohMeaAlt}
    \widetilde C(\rho, \{\Pi_a\}) \coloneqq \frac{C^{NCl}_{KD}(\rho, \{\Pi_a\})}{\max_{\rho} C^{NCl}_{KD}(\rho, \{\Pi_a\})} = \frac{C^{NCl}_{KD}(\rho, \{\Pi_a\})}{\sqrt{d} - 1},
\end{equation}
where we have used $\max_{\rho} C^{NCl}_{KD}(\rho, \{\Pi_a\}) = \max_{\ket{\psi}} C^{NCl}_{KD}(\ket{\psi}, \{\Pi_a\}) = \sqrt{d} - 1$.
The values of the normalised mixed-state coherence $\widetilde C(\rho, \{\Pi_a\})$ for a qubit with respect to the polar angle $\theta$ and different radii are shown alongside those of the $\ell_1$-coherence $C_{l_1}(\rho, \{\Pi_a\})$ in Fig. \ref{fig:cohplot}.
In order for it to be a proper coherence monotone, it needs to satisfy the following properties \cite{Budiyono:2023riv, Budiyono:2023vcn} (for proofs see App. \ref{app: ProCoh}): (1) Convexity, which is automatically ensured by the nature of convex–roof construction $ C(\rho,\{\Pi_a\}) \leq \sum_i p_i C(\rho_i,\{\Pi_a\}) $ for $\rho = \sum_i p_i \rho_i$; (2) Faithfulness, where the coherence is maximum when the target is in the set of maximally coherent states for a given basis $\{\Pi_a\}$, $\ket{\psi_{mc}} = \frac{1}{\sqrt{d}} \sum_a e^{i\theta_a} \ket{a}$, then $C(\ket{\psi_{mc}},\{\Pi_a\}) = \max C(\ket{\psi},\{\Pi_a\})$ and minimum when it is in the set of incoherent states, $\rho_I = \sum_a p_a \Pi_a$, then $C(\rho_I,\{\Pi_a\}) = 0$; (3) Non-increasing under a decoherence operation $\Theta_\lambda (\rho, \{ \Pi_a \}) \coloneqq \lambda \rho + (1-\lambda) \Delta_a(\rho)$, where $ C(\Theta_\lambda(\rho,\{\Pi_a\})) \leq C(\rho,\{\Pi_a\}) $; (4) Non-increasing under partial trace, where $ C(\rho_{12}, \{\Pi_a \otimes I\}) \geq C(\rho_{1}, \{\Pi_a\}) $ with $\rho_1 = \text{Tr}_2~ \rho_{12} $; (5) Unitarily covariant, where $ C(U\rho U^\dagger, \{U \Pi_a U^\dagger\}) = C(\rho, \{ \Pi_a \}) $; (6) Invariant under unitary transformation that commutes with a Hermitian operator $A$ that has $\{\Pi_a\}$ as the eigenvectors (i.e. $A = \sum_a a \ket{a}\bra{a}$ with $a \in \mathbb{R}$), i.e. $ C(U\rho U^\dagger, \{ \Pi_a \}) = C(\rho, \{ \Pi_a \}) $ with $[A, U] = 0$ and (7) Invariant under permutation of incoherent basis.

\begin{figure}[!tbp]
    \centering
    \includegraphics[scale=0.55]{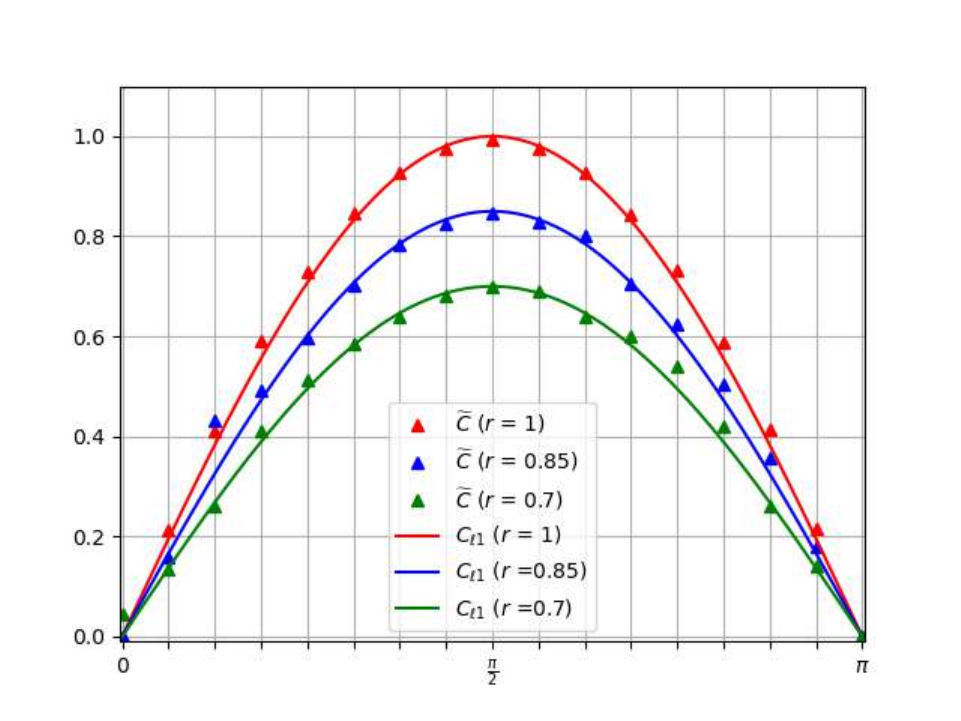}
    \caption{\justifying Plots of the normalised coherence $\widetilde{C}$ and $\ell_1$-coherence $C_{\ell_1}$ of a single qubit as a function of $\theta$ for radii of the density matrices $r = 1, 0.85, 0.7$. The coherence plot is obtained by solving the associated convex optimisation problem using \texttt{scipy.optimize.minimize()} in Python.}
    \label{fig:cohplot}
\end{figure}

\section{Complementarity Relation between  Predictability and Coherence}
\label{sec: CRBurConv}

\subsection{Pure States}
\label{subsec: ConRoo}
 Following the definition in Eq. \eqref{eq: BurDisPreMea}, Eq. \eqref{eq: KDCohPur} can also be written in terms of the  predictability,
\begin{equation}
    \label{eq: KDCohPurBur}
    \widetilde C(\ket{\psi}, \{\Pi_a\}) = 1 - \widetilde P(\ket{\psi}, \{\Pi_a\}).
\end{equation}
In order to obtain a geometrical meaning, first, we
consider a set of the states $\ket{\psi_{mc}}$ that contain the most coherence states with respect to $\{ \Pi_a \}$ 
, as opposed to the set of incoherent states of $\{\Pi_a\}$, $\rho_I$
. This leads to having the ability to write $I/d$ in terms of the set $\Delta_a(\ket{\psi_{mc}}) = I/d$. Thus Eq. \eqref{eq: KDCohPurBur} can also be written explicitly as
\begin{eqnarray}
    \label{eq: KDCohPurBurEle}
    &&\frac{C^{NCl}_{KD}(\ket{\psi}, \{\Pi_a\})}{\max_{\ket{\psi}} C^{NCl}_{KD}(\ket{\psi}, \{\Pi_a\})}\nonumber\\
    &&= 1 - \frac{D_B^2 (\Delta_a(\ket{\psi}) \| \Delta_a(\ket{\psi_{mc}}))}{\max_{\ket{\psi}} D_B^2(\Delta_a(\ket{\psi}) \| \Delta_a(\ket{\psi_{mc}}))}.
\end{eqnarray}
This relation provides a quantitative expression of Bohr's uncertainty
or the wave--particle duality relation for arbitrary dimensional quantum systems in the case of pure states. It shows how the existence of coherence complements and limits the which-way information to predict with certainty the outcomes of our system. Furthermore, the inclusion of the dephasing operation in the expression of coherence avoids the need to choose which state in the incoherent-state set as the reference state, as illustrated in Fig. \ref{fig:cohvis}.

\begin{figure}[!bp]
        \centering
        \includegraphics[scale=0.15]{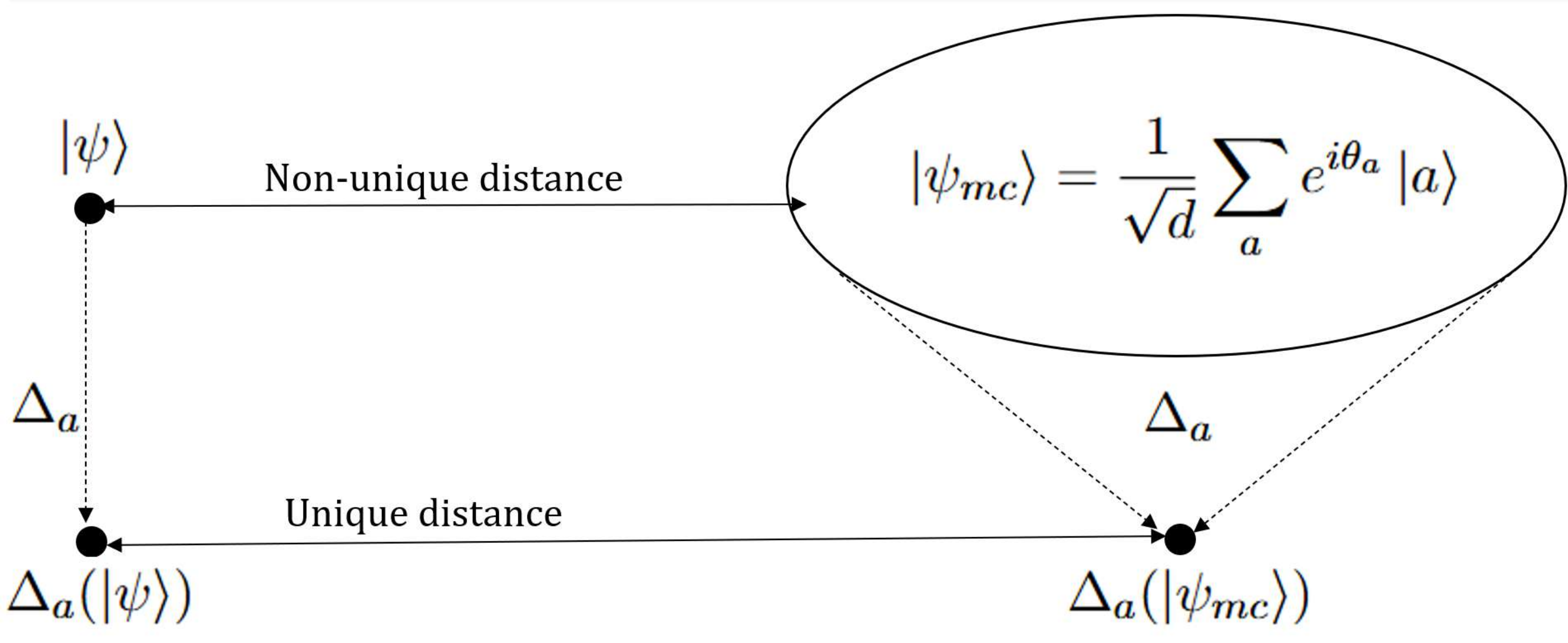}
        \caption{\justifying This figure shows the geometric interpretation of defining the coherence monotone as a distance between the dephased target state $\Delta_a(\ket{\psi})$ and the dephased maximally coherent state $\Delta_a(\ket{\psi_{mc}})$. The measure between the target state $\ket{\psi}$ and the set of maximally coherent states is not unique since there are infinitely many states within the set (as seen from the parameters $\theta_a$). However, the measure between the dephasing channels of both references is unique.}
        \label{fig:cohvis}
\end{figure}

\subsection{Mixed States}
\label{subsec: NKDWPD}

From this form of coherence, we can intrinsically derive 
a wave–particle duality relation for the general case of mixed states via \textit{convex-roof construction}. 
It has the form expressed below,
\begin{equation}
    \label{eq: WavParBurCoh}
    \widetilde{C}(\rho, \{\Pi_a\}) + \widetilde{P}(\rho, \{\Pi_a\}) \leq 1,
\end{equation}
where the normalised coherence $\widetilde C (\rho, \{\Pi_a\})$ represents wave-like character, while the normalised predictability represents particle-like character.
The proof is contained in App. \ref{app: Wpd}.

Eq. \eqref{eq: WavParBurCoh} has the same complementarity structure as the conventional wave--particle duality relations in interferometry, where a particle-like quantifier (e.g.\ path predictability or which-way distinguishability) trades off against a wave-like quantifier (e.g.\ fringe visibility) \cite{Wootters:1979, Greenberger:1988aih, Englert:1996zz}.
In our measurement-centric formulation, the particle term is the basis predictability of the dephased statistics, quantified geometrically via the Bures distance \(P(\rho,\{\Pi_a\}) \) as in Sec. \ref{subsec: BurDisPreMea}, while the wave term is a basis-relative coherence constructed from the KD coherence and extended to mixed states by the convex roof as in Sec. \ref{subsec: PreNCLKD} \cite{Bera:2015wbe}. For pure states, these two notions are tightly linked: the normalised KD coherence is exactly the complement of the normalised predictability, cf.\ Eq. \eqref{eq: KDCohPurBur}, so the duality is saturated in that case. For mixed states, the convex-roof extension optimises over classical decompositions, and the resulting trade-off becomes the inequality in Eq. \eqref{eq: WavParBurCoh}, mirroring the standard observation that classical mixing and decoherence generally prevent simultaneous maximisation of wave- and particle-like behaviour.

\section{Operational Interpretation}
\label{sec: OpInt}

\subsection{Coherence as Classical Irreducible Randomness}
\label{subsec: QuaGame}

The mixed-state trade-off derived above suggests a direct operational reading of the convex-roof coherence. Instead of viewing the optimisation only as a mathematical construction, we may interpret each pure-state decomposition of $\rho$ as a possible classical preparation strategy, where the decomposition label is available as side information. In this setting, the coherence quantifies the part of the measurement uncertainty that remains after the best classical explanation of the source has been taken into account.

The previous wave--particle duality relation can be alternatively rewritten as
\begin{equation}
    \label{eq: PreCohInt}
    \frac{C(\rho, \{\Pi_a\})}{\max_{\rho} C(\rho, \{\Pi_a\})}  =  1 - \sup_{\{p_i, \ket{\psi_i}\}} \sum_{i} p_i \frac{P(\ket{\psi_i}, \{\Pi_a\})}{\max_{\ket{\psi}} P(\ket{\psi}, \{\Pi_a\})}.
\end{equation}
Written in this form, the wave--particle duality relation admits an operational setting with a procedure in analogy with the \textit{quantum guessing task with posterior information}~\cite{Carmeli:2021jho} and \textit{assisted coherence distillation}~\cite{Zhao:2021org}, which gives a statistical interpretation of the convex-roof construction approach. The procedure is as follows.

\begin{enumerate}[label=(\roman*), leftmargin=0pt, itemindent=*, align=left]
    \item \textbf{Alice chooses a decomposition.}
    Alice wants to communicate a quantum state $\rho$, but she does not transmit
    $\rho$ directly. Instead, she chooses a pure-state ensemble
        $\rho = \sum_i p_i \ket{\psi_i}\bra{\psi_i}$,
    where the classical label $i$ is chosen with probability $p_i$.

    \item \textbf{Alice sends a labelled pure state.}
    In each round, Alice samples a label $i$ according to $p_i$ and transmits the
    corresponding pure state $\ket{\psi_i}$ to Bob. The label $i$ is treated as
    classical side information associated with the preparation.

    \item \textbf{Bob measures in a fixed basis.}
    Bob receives $\ket{\psi_i}$ and performs the fixed projective measurement
    $\{\Pi_a\}$. The outcome $a$ is obtained with probability
        $p(a|i) = \Tr\!\left[\Pi_a \ket{\psi_i}\bra{\psi_i}\right]$.
    Thus, the measurement probes how definite the transmitted state is with
    respect to the basis $\{\Pi_a\}$.

    \item \textbf{Bob uses the label to assess predictability.}
    Given the classical label $i$, Bob knows which label was transmitted.
    The certainty of the measurement outcome for this round is quantified by
    the predictability
       $ P(\ket{\psi_i},\{\Pi_a\})$.
    This quantity is maximal when the measurement outcome is certain, and
    vanishes when the outcome is no better than a random guess in the chosen
    basis.

    \item \textbf{Bob evaluates the average certainty of the ensemble.}
    For a fixed decomposition of $\rho$, the overall predictability available to
    Bob is the average
        $\sum_i p_i\,P(\ket{\psi_i},\{\Pi_a\})$.
    This measures how much of the measurement randomness can be removed once
    the classical preparation label is known.

    \item \textbf{Alice and Bob optimise over all decompositions.}
    Since the same mixed state $\rho$ may admit many pure-state decompositions,
    Alice and Bob may choose the decomposition that makes the average
    predictability as large as possible:
        $\sup_{\{p_i,\ket{\psi_i}\}}
        \sum_i p_i\,P(\ket{\psi_i},\{\Pi_a\})$.
\end{enumerate}

From the above illustration, our quantity operationally measures the part of the measurement unpredictability that cannot be removed by conditioning on any classical label specifying a pure-state decomposition of $\rho$. We refer to this residual unpredictability as \textit{classically irreducible randomness}. This is similar to the concept of \textit{intrinsic randomness}, where it is generated purely by the quantum nature of the state, as opposed to \textit{extrinsic randomness}, which can be removed by knowing hidden variables and classical conditioning~\cite{Dai:2022feb}.

The coherence, as shown in Eq. \eqref{eq: PreCohInt}, is expressed as the complement of the maximum of the average post-processed predictability. This leads to an interpretation that
follows the physical intuitions for the two extreme cases: \textit{maximally coherent states} and \textit{incoherent states}. For maximally coherent states, one starts with a maximally coherent (with respect to the measurement basis) quantum state $\rho_{mc}$, which is also a pure state, and pure states cannot be decomposed non-trivially. Thus, there is only one possible way of transmitting the state. Since it is maximally coherent, there is no hint of the state in the chosen basis and the outcome can only be guessed randomly. Meanwhile, for an incoherent state $\rho_I$, a decomposition can be chosen such that the pure states lie within the measurement basis. With this information, the measurement outcome can be certain for each transmission. Looking at it another way, turning on coherence with respect to a certain basis will disturb one's ability to predict the state during measurement using the same basis.

To illustrate the interpretation, consider the following examples for a single qubit with measurement in the computational basis. According to Carath\'eodory's theorem, for a single qubit, generally, at most four pure states are needed in the decomposition, however, for illustrative purposes, we will consider two pure states.
\begin{enumerate}[wide, labelwidth=!, labelindent=0pt]
    \item The maximally mixed state $\rho = I/2$

    This is the classically maximal random state. Let the measurement be performed in the computational basis. One possible pure-state decomposition is
    \begin{equation*}
        \rho = \frac{1}{2}\ket{+}\bra{+} + \frac{1}{2}\ket{-}\bra{-},
    \end{equation*}
    where $\ket{\pm} = \frac{1}{\sqrt{2}} \left( \ket{0} \pm \ket{1} \right)$. If this decomposition is chosen, the average predictability with respect to the computational basis is zero because $\ket{+}$ and $\ket{-}$ are maximally unpredictable with respect to $\ket{0}$ and $\ket{1}$, i.e.
    \begin{equation*}
        P(\ket{\pm}, \{\ket{0},\ket{1}\}) = 0.
    \end{equation*}
    This means the resulting ensemble is fully random. However, this decomposition hides the classical origin of randomness. The same density matrix admits the alternative decomposition
    \begin{equation*}
        \rho = \frac{1}{2}\ket{0}\bra{0} + \frac{1}{2}\ket{1}\bra{1}.
    \end{equation*}
    If this decomposition is instead prepared, the outcome is perfectly predictable for every run, i.e.
    \begin{equation*}
        P(\ket{0(1)}, \{\ket{0},\ket{1}\}) = P_{max} = 2 \left( 1-\frac{1}{\sqrt{2}} \right),
    \end{equation*}
    which leads to the average predictability attaining its maximum value. This shows that the apparent randomness of the maximally mixed state is not quantum, as it can be \textit{classically reduced} by choosing an appropriate pure-state decomposition (and revealing the corresponding classical label). In our framework, this is reflected by 
    \begin{equation*}
        C(\rho, \{\ket{0},\ket{1}\}) = 0,
    \end{equation*}
    since the supremum of the average predictability is reached by decomposing the state into $\ket{0}$ and $\ket{1}$.
    
    \item The maximally coherent state $\rho = \ket{+}\bra{+}$

    The resulting average predictability of this explicit expression is zero, and, since for a pure state there is no non-trivial decomposition, there is no alternative decomposition to reduce and `explain away' the randomness by classical side information. Thus, in this sense, the randomness is purely generated by the quantumness of the state and is \textit{classically irreducible}. In our framework, this becomes
    \begin{equation*}
        C(\rho, \{\ket{0},\ket{1}\}) = C_{max} = \sqrt{2}-1,
    \end{equation*}
    since the supremum of the average predictability remains zero for this state.

    \item Consider the following state,

    \begin{equation*}
        \rho = \begin{pmatrix}
            3/4 & 1/4\\
            1/4 & 1/4
        \end{pmatrix}.
    \end{equation*}

    This state can be decomposed as
    \begin{equation*}
        \rho = \frac{1}{2} \ket{\psi_+} \bra{\psi_+} + \frac{1}{2} \ket{\psi_-} \bra{\psi_-},
    \end{equation*}
    where $\ket{\psi_\pm} = \frac{\sqrt{3}}{2} \ket{0} + \frac{1}{2} e^{\pm i\arctan\sqrt{2}} \ket{1}$. Note that $1 - P(\ket{\psi_+}) / P_{max} = 1 - P(\ket{\psi_-}) / P_{max} = 0.366$, thus giving the average unpredictability $1 - P_{avg} = 0.366$. Another way to decompose it is
    \begin{equation*}
        \rho = \frac{1}{2} \ket{+}\bra{+} + \frac{1}{2} \ket{0}\bra{0},
    \end{equation*}
    which gives $1 - P_{avg} / P_{max} = 0.207$. This illustrates two key aspects of the convex-roof construction. First, the same mixed state can be realised by ensembles with very different uses of classical side information: in the first decomposition, the label $\{\ket{\psi_\pm}\}$ is less useful for predicting the measurement outcomes since 
    the normalised predictability complement is larger than in the latter decomposition. The label $\{\ket{+}, \ket{0}\}$ explains the apparent randomness by identifying runs that are predictable $\{\ket{0}\}$ and those that are maximally unpredictable $\{\ket{+}\}$. Second, because our mixed-state coherence is the infimum of the average of the pure-state coherence (i.e.\ unpredictability) over all decompositions, the latter provides a tighter bound. In this operational picture, by choosing a more favourable ensemble, the average predictability can be increased, therefore reducing the classically irreducible part of the randomness quantified by $C(\rho, \{\Pi_a\})$.
\end{enumerate}

\subsection{Coherence as an Unpredictability Resource for Random Generation}
\label{subsec: QRNG}

Quantum random-number generation (QRNG) aims to certify randomness from intrinsically quantum measurement outcomes rather than from classical ignorance~\cite{Rarity:1994, Stefanov:1999ue, Jennewein:2000, Ma:2013, Ma:2016}. In source-independent quantum random-number generation (SI QRNG), randomness must remain certifiable even when the source is untrusted~~\cite{Cao:2016vsg, Marangon:2017sdi, Avesani:2018sdi, Michel:2019si}. Our construction captures the worst-case setting by allowing the strongest classical pure-state decomposition of $\rho$, leading to a min-entropy bound that survives source-side information.

To make this connection precise, let $\rho$ be an arbitrary $d$-dimensional quantum state and let $\{{\Pi_a}\}$ be a fixed set of projective measurement operators, with outcome probabilities $p_a = \operatorname{Tr}(\rho\Pi_a)$.
We define the variable for describing the optimal probability of single-shot guessing the measurement outcome $A$ of the state $\rho$ with respect to a set of measurement basis $\{\Pi_a\}$ as $\text{Pr}_{\text{guess}} (A)_\rho \coloneqq \max_a p_a$, which is called the \textit{optimal guessing probability}~\cite{Massey:1994, Boztas:2014, Rioul:2022}.
From this quantity, we can derive the min-entropy $H_{min}$, which quantifies the randomness in measurement outcomes, and thus acts as an important cryptographic quantity, especially in random-number generation~\cite{Koenig:2009avh, Ma:2013, Ma:2016, Herrero-Collantes:2016lbz, Turan:2018}. The min-entropy is defined as
\begin{equation}
    H_{min}(\text{Pr}(A)_\rho) \coloneqq H_{\infty} (\text{Pr}(A)_\rho) = -\log \text{Pr}_{\text{guess}} (A)_\rho.
\end{equation}
Here $H_\infty$ denotes the $\alpha\to\infty$ limit of the Rényi entropy
$H_\alpha=\frac{1}{1-\alpha}\log\sum_a \Pr(A)_\rho^\alpha$~\cite{Koenig:2009avh, Boztas:2014}.

\begin{theorem}[Worst-case classical-label source-independent min-entropy bound]
\label{thm:classical-label-source-side-bound}
Following the standard definitions of the classical conditional guessing probability
and min-entropy~\cite{Konig2009Operational, Tomamichel2017FiniteKeyLimits}, let $\rho$ be an arbitrary $d$-dimensional quantum state and let
$\{\Pi_a\}$ be a fixed projective measurement basis. Consider an
untrusted source that may realise $\rho$ through any pure-state ensemble
decomposition
$\rho=\sum_i p_i\ket{\psi_i}\bra{\psi_i}$.
If the classical preparation label $i$ is available as side information,
then the relevant label-assisted guessing probability of the measurement outcome $A$, given the classical preparation label $I$, is
\begin{equation}
\label{eq: clOptPr}
    \text{Pr}_{\rm guess}^{\rm cl}(A|I)_\rho
    \coloneqq
    \sup_{\{p_i,\ket{\psi_i}\}}
    \sum_i p_i \max_a
    \bra{\psi_i}\Pi_a\ket{\psi_i},
\end{equation}
where the supremum is taken over all pure-state decompositions of $\rho$.
The corresponding classical-label conditional min-entropy is
\begin{equation}
\label{eq: clMinEnt}
    H_{\min}^{\rm cl}(A|I)_\rho
    \coloneqq
    -\log \text{Pr}_{\rm guess}^{\rm cl}(A|I)_\rho .
\end{equation}
Then
\begin{eqnarray}
    \label{eq: Th1}
    H_{\min}^{\rm cl}(A|I)_\rho 
    &\geq& 
    H_{\min,\mathrm{LB}}^{\rm cl}(\rho,\{\Pi_a\})\nonumber\\
    &\coloneqq&
    -\log Q_d\!\left[
    1+C(\rho,\{\Pi_a\})
    \right],
\end{eqnarray}
where
\begin{equation}
    Q_d(s) = \left[ \frac{s+\sqrt{(d-s^2)(d-1)}}{d} \right]^2.
\end{equation}
This bound is valid for every state $\rho$. In particular, it remains valid
in the worst-case classical source model where the source is allowed to
choose the pure-state decomposition of $\rho$ that maximises the
label-assisted guessing probability.
\end{theorem}
The proof of the theorem can be found in App.~\ref{app: Unp}.

The min-entropy $H^{cl}_{\min}$ admits a lower bound $H^{cl}_{\min,\mathrm{LB}}$ where $-\log Q_d$ is monotonically increasing in the coherence, as shown in Fig. \ref{fig:cohunp}. Consequently, higher coherence yields a larger certified amount of min-entropy, strengthening the randomness available for random-number generation~\cite{Cao:2016vsg}. In contrast, increasing the  predictability decreases this lower bound, reflecting the intuition that greater predictability reduces the guaranteed randomness of the system.

\begin{figure}[!tbp]
    \centering
    \includegraphics[scale=0.6]{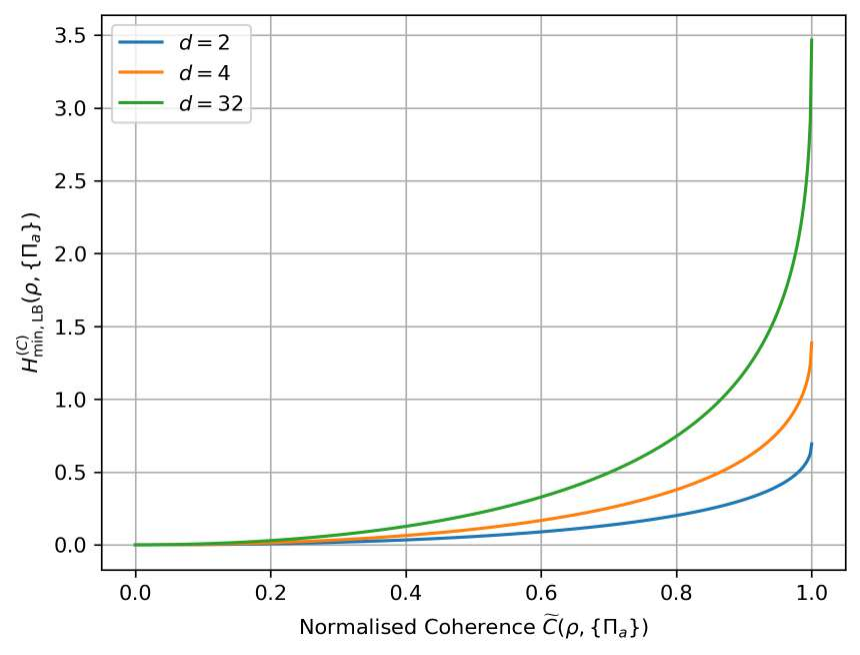}
    \caption{\justifying The lower bound of the min-entropy $H_{\mathrm{min},\mathrm{LB}}^{(C)}$ versus the normalised coherence $\widetilde{C}$, which shows that the function is monotonically increasing with respect to coherence $C$ as in the inequality \eqref{eq: Th1} and dimension $d$.}
    \label{fig:cohunp}
\end{figure}

In the present framework, 
the source can be represented as an ensemble of pure states. The decomposition label $i$ acts as classical side information in the procedure. Since guessing the outcome corresponds to selecting the most likely projector $\Pi_a$ in the fixed measurement basis $\{\Pi_a\}$, as described in Sec. \ref{subsec: QuaGame}, we are interested in the quantity $\max_a \bra{\psi_i}\Pi_a \ket{\psi_i}$ in Eq. \eqref{eq: clOptPr}.
The most conservative classical source model is therefore obtained by optimising the summation of these quantities over all ensembles compatible with $\rho$. Our convex-roof construction in Eq. \eqref{eq: PreCohInt} follows the same adversarial logic, but with the  predictability as the figure of merit.

Equivalently, one may imagine a powerful source-side agent who is able to decompose $\rho$ into the most favourable ensemble of pure states, chosen so as to maximise the average predictability of the measurement outcomes. Thus, the coherence quantifies the part of the measurement unpredictability that remains after the most favourable classical decomposition label has been revealed. The resulting min-entropy bound is therefore a worst-case bound: even after this optimal classical explanation of the source has been allowed, the remaining randomness is still certified by the existence of the coherence. This gives the measure an SI QRNG interpretation against classical source-side information: a large value of $C(\rho,\{\Pi_a\})$ certifies that the observed randomness cannot be explained purely by a hidden classical preparation label. This complements existing
coherence-based approaches to source-independent QRNG ~\cite{Ma:2019aol}.

This should be distinguished from a full composable security proof for SI QRNG, where the adversary may hold arbitrary quantum side information $E$ and the relevant quantity is the conditional min-entropy $H_{\min}(A|E)$. The present result addresses the classical-label layer of this problem, corresponding to side information $I$ about the preparation ensemble. Extending the present convex-roof construction to bounds on $H_{\min}(A|E)$ is a natural direction for future work.

\section{Summary and Outlook}

We developed a measurement-centred formulation of wave--particle complementarity for multi-path quantum systems. Starting from a fixed measurement basis, we defined the particle-like contribution through a predictability measure, given by the squared Bures distance between the dephased state and the maximally mixed state. Since the dephasing removes all phase information while preserving the outcome probabilities, this quantity depends only on the observed basis statistics. It therefore provides a geometric and experimentally accessible measure of how strongly the state favours some paths over others.

We then showed that this predictability is naturally paired with a basis-dependent coherence measure. For pure states, the normalised  predictability is exactly complementary to the normalised non-classical Kirkwood--Dirac coherence, giving a saturated wave--particle duality relation. For mixed states, extending the pure-state coherence by the convex roof leads to a corresponding inequality between predictability and coherence. In this way, the same Bures/Bhattacharyya geometry gives a unified description of the particle-like and wave-like aspects of the state: predictability measures definiteness in the chosen basis, while coherence measures the remaining quantum superposition relative to that basis.

The convex-roof construction also gives the coherence measure a direct operational meaning. It quantifies the part of the measurement unpredictability that remains even after one optimises over all pure-state decompositions of the density matrix and reveals the corresponding classical preparation label. This residual contribution can therefore be interpreted as classically irreducible randomness: randomness that cannot be explained away by classical side information about the preparation. This interpretation leads naturally to bounds on the optimal guessing probability and on the min-entropy of the measurement outcomes. In particular, it gives the coherence measure a source-independent QRNG interpretation against classical source-side information, where the remaining randomness cannot be explained purely by a hidden classical preparation label.

Several directions remain open for future work. First, it would be interesting to connect the present framework more directly to relative entropy and free-energy based formulations, in order to explore whether the predictability--coherence trade-off can be extended into a triadic relation involving a thermodynamic contribution. Such a perspective may clarify the role of quantum thermodynamics in constraining the interplay between classical definiteness and quantum superposition.
Second, it would be natural to investigate extensions involving entanglement, especially in multipartite settings where local predictability and local coherence may be limited by nonlocal quantum correlations. This may reveal a broader complementarity structure linking predictability, coherence, and entanglement.
Third, the present analysis could be generalised by explicitly incorporating detector states, in the spirit of which-path detection scenarios, so as to compare the present measurement-centred formulation with distinguishability-based approaches to wave--particle duality.
Finally, it would be worthwhile to explore possible experimental applications, particularly in quantum-optical platforms such as multi-path interferometers, integrated photonic circuits, and coherence-based random-number-generation protocols, where the quantities introduced here may be directly accessible from measurement statistics.

\acknowledgments{We would like to thank Agung Budiyono for the discussions and ideas for this research. 

%
%
%
\bibliographystyle{JHEP} 
\bibliography{refs_cleaned} 
\clearpage  
\onecolumngrid
\input{appendix}
\end{document}

%% file: appendix.tex
\section{Appendix}

\subsection{Proof for monotonicity under mixing for the Predictability}
\label{app: Pro56}

This property is proven by acknowledging the fact that the fidelity function is concave and using Jensen's inequality,
\begin{proof}
    Starting from the property of concavity for the fidelity~\cite{Nielsen:2012yss}.
    \begin{equation}
        F \left(\sum_i p_i \rho_i, \sigma \right) \geq \sum_i p_i  F \left( \rho_i, \sigma \right).
    \end{equation}
    The square root of the fidelity function also follows this property
    \begin{eqnarray}
        \sqrt{F \left(\sum_i p_i \rho_i, \sigma \right)} &\geq& \sqrt{\sum_i p_i  F \left( \rho_i, \sigma \right)}\nonumber\\
        &\geq& \sum_i p_i  \sqrt{F \left( \rho_i, \sigma \right)},\nonumber\\
        - \sqrt{F \left(\sum_i p_i \rho_i, \sigma \right)} &\leq& - \sum_i p_i  \sqrt{F \left( \rho_i, \sigma \right)}, 
    \end{eqnarray}
    where we have used the Jensen's inequality for the square root function, which is a concave function, for the second line.
    Translating this to the Bures distance function, we have
    \begin{eqnarray}
        D_B^2 \left( \sum_i p_i \rho_i \| \sigma\right) &=& 2 \left( 1 - \sqrt{F \left(\sum_i p_i \rho_i, \sigma \right)} \right)\nonumber\\
        &\leq& 2 \left( 1 - \sum_i p_i  \sqrt{F \left( \rho_i, \sigma \right)} \right)\nonumber\\
        &=& \sum_i p_i  ~2\left( 1 - \sqrt{F \left( \rho_i, \sigma \right)} \right)\nonumber\\
        &\leq& \sum_i p_i D_B^2 \left( \rho_i \| \sigma\right).
    \end{eqnarray}
Using the previous result for depolarising operation,
    \begin{eqnarray}
        P\bigl(\Lambda_\lambda(\rho)\bigr) &\leq& (1-\lambda) P(\rho) + \lambda P(I/d)\nonumber\\
        &=& (1-\lambda) P(\rho)\nonumber\\
        &\leq& P(\rho),
    \end{eqnarray}
    where we have used the fact that $P(I/d) = 0$ since it is a maximally unpredictable state.
\end{proof}

\subsection{Proof for the properties of the coherence}
\label{app: ProCoh}

In this section, we prove the properties laid out in subsection \ref{subsec: ConRoo}. For practical purposes, we write the maximum values $D_B^2(\Delta_a(\ket{\psi}) \| \Delta_a(\ket{\psi_{mc}}))$ as $D$ which has the value $D = 2\left( 1-\frac{1}{\sqrt{d}} \right)$ and the maximum value of $C(\rho, \{ \Pi_a \})$ as $C = \sqrt{d}-1$.

\subsubsection{Convexity}
\label{appsub: prop01}
    If the mixed state is decomposed into pure states (i.e. $\rho = \sum_i p_i \ket{\psi_i}\bra{\psi_i}$), then, by definition of convex–roof construction, the resulting construction must satisfy convexity since it's the infimum of such a decomposition. We provide the proof that the convex–roof construction guarantees the convexity of the coherence by its nature.
    \begin{proof}
        \label{pr: EnsConv}
        To verify convexity for mixed-state decomposition, consider the following mixed quantum state that can be decomposed into another two mixed states,
        \begin{equation}
            \label{eq: EnsConv01}
            \rho = p_1 \rho_1 + p_2 \rho_2.
        \end{equation}
        Say that $\rho_1$ and $\rho_2$ have the following optimal decompositions,
        \begin{eqnarray}
            \label{eq: EnsConv02}
            \rho_1 &=& \sum_i a_i \ket{\psi_i}\bra{\psi_i},\\
            \rho_2 &=& \sum_i b_i \ket{\phi_i}\bra{\phi_i},
        \end{eqnarray}
        The coherence, then, has the following form of its convex–roof construction,
        \begin{eqnarray}
            \label{eq: EnsConv03}
            C(\rho_1,\{\Pi_a\}) &=& \sum_i a_i C(\ket{\psi_i},\{\Pi_a\}),\\
            C(\rho_2,\{\Pi_a\}) &=& \sum_i b_i C(\ket{\phi_i},\{\Pi_a\}).
        \end{eqnarray}
        Going back to eq. \eqref{eq: EnsConv01}; we can write it as
        \begin{equation}
            \label{eq: ENsConv04}
            \rho = p_1 \sum_i a_i \ket{\psi_i}\bra{\psi_i} + p_2 \sum_i b_i \ket{\phi_i}\bra{\phi_i},
        \end{equation}
        which is a pure-state decomposition of $\rho$. However, according to the convex–roof construction, the coherence of $\rho$ is expressed as
        \begin{eqnarray}
        \label{eq: EnsConv05}
            C(\rho,\{\Pi_a\}) &\leq& p_1 \sum_i a_i C(\ket{\psi_i},\{\Pi_a\})\nonumber\\
            &&+ p_2 \sum_i b_i C(\ket{\phi_i},\{\Pi_a\})\nonumber\\
            &=& p_1 C(\rho_1,\{\Pi_a\}) + p_2 C(\rho_2,\{\Pi_a\}).
        \end{eqnarray}
        This can be easily generalised to a decomposition of $\rho$ with an arbitrary number of terms.
    \end{proof}

\subsubsection{Faithfulness}
\label{appsub: prop02}
    \begin{proof}

    In order to verify the faithfulness of our coherence, we apply it on two sets: the set of maximally coherent states $\ket{\psi_{mc}}$, which gives the maximum value, and the set of incoherent states $\rho_I$, which gives the minimum value.

    Finding the maximum value is straightforward since the set consists of pure states, thus we just need to apply the formula for pure states in eq. \eqref{eq: KDCohPurBurEle},
    \begin{eqnarray}
        \frac{C^{NCl}_{KD}(\ket{\psi_{mc}}, \{\Pi_a\})}{C} &=& 1 - \frac{D_B^2(\Delta_a(\ket{\psi_{mc}}) \| \Delta_a(\ket{\psi_{mc}}))}{D}\nonumber\\
        &=& 1,
    \end{eqnarray}
    where we have used the fact that the Bures distance between two equal states is zero. The resulting value is the maximum value of the coherence.
    
    Next, we examine the set of incoherent states $\rho_I$.
    The coherence can also be written as the following,
    \begin{eqnarray}
        \label{der: CohFai1}
        C(\rho_I, \{\Pi_a\}) 
        &=& \inf_{\{p_i, \ket{\psi_i}\}} \sum_{i} p_i C(\ket{\psi_i}, \{\Pi_a\})\nonumber\\
        &=& \inf_{\{p_i, \ket{\psi_i}\}} \sum_{i} p_i \left\{ 1 - \frac{1}{D} D_B^2(\Delta_a(\ket{\psi_i}) \| \Delta_a(\ket{\psi_{mc}}) ) \right\}\nonumber\\
        &=&1 - \frac{1}{D} \sup_{\{p_i, \ket{\psi_i}\}} \sum_{i} p_i  D_B^2(\Delta_a(\ket{\psi_i}) \| \Delta_a(\ket{\psi_{mc}}) ).\nonumber\\
    \end{eqnarray}
    In this form of the coherence, we have to find the maximum value of the decomposition of the Bures distance. By using the exact form of the incoherent state $\rho_I = \sum_a p_a \Pi_a$ as the decomposition for the convex–roof construction of the coherence, we get
    \begin{eqnarray}
        \label{der: CohFai2}
        \sum_i p_i D_B^2(\Delta_a(\ket{\psi_i}) \| \Delta_a(\ket{\psi_{mc}}) )
        &=& \sum_b p_b D_B^2(\Delta_a(\Pi_b) \| \Delta_a(\ket{\psi_{mc}}) )\nonumber\\
        &=& \left(\sum_b p_b \right) D_B^2(\Delta_a(\Pi_b) \| \Delta_a(\ket{\psi_{mc}}) )\nonumber\\
        &=& D_B^2(\Delta_a(\Pi_b) \| \Delta_a(\ket{\psi_{mc}}) )\nonumber\\
        &=& D,
    \end{eqnarray}
    where we have used the fact that the Bures distance for $\Pi_b \in \{\Pi_a\}$ is the same. This is the maximum value of the Bures distance in the numerator. Thus, we can use this decomposition for the incoherent quantum state,
    \begin{eqnarray}
        \label{der: CohFai3}
        \frac{C(\rho_I, \{\Pi_a\})}{C} &=&  1 - \frac{1}{D} D \nonumber\\
        &=& 0.
    \end{eqnarray}
    So, the coherence of the incoherent states is zero.

    For the last part of the proof, we verify the states that give maximum coherence (i.e. $C(\rho_1,\{\Pi_a\}) = C$) and zero coherence (i.e. $C(\rho_0,\{\Pi_a\}) = 0$). The condition for giving maximum coherence is for
    \begin{equation}
        \label{eq: CohFai4}
        \sup_{\{p_i, \ket{\psi_i}\}} \sum_{i} p_i  D_B^2(\Delta_a(\ket{\psi_i}) \| \Delta_a(\ket{\psi_{mc}}) ) = 0.
    \end{equation}
    This happens solely when $D_B^2(\Delta_a(\ket{\psi_i}) \| \Delta_a(\ket{\psi_{mc}}) ) = 0$ which is if $\ket{\psi_i} $ is a maximally coherent state with respect to $\{ \Pi_a \}$.    
    Likewise, the condition for zero coherence is
    \begin{equation}
        \label{eq: CohFai5}
        \sup_{\{p_i, \ket{\psi_i}\}} \sum_{i} p_i  D_B^2(\Delta_a(\ket{\psi_i}) \| \Delta_a(\ket{\psi_{mc}}) ) = D.
    \end{equation}
    This only holds if $D_B^2(\Delta_a(\ket{\psi_i}) \| \Delta_a(\ket{\psi_{mc}}) = D$ (i.e. maximum Bures distance). The only pure states that have this value are those that lie on an incoherent basis:
    \begin{equation}
        \label{eq: CohFai6}
        \ket{\psi_i}\bra{\psi_i} \in \{ \Pi_a \}.
    \end{equation}
    The resulting quantum state is inside the incoherent set,
    \begin{equation}
        \label{der: CohFai6}
        \rho = \sum_{a} p_a \ket{a}\bra{a} \in \rho_I.
    \end{equation}
    \end{proof}

\subsubsection{Non-increasing under decoherence operation}
\label{appsub: prop03}

    For a decoherence operation, we find it to exhibit non-increasing monotonicity.
    \begin{proof}
        We consider the coherence of a state after undergoing a decoherence operation.
        \begin{eqnarray}
            \label{der: DecMon}
            C(\Theta_\lambda(\rho,\{\Pi_a\})) &\leq&  \lambda C(\rho, \{ \Pi_a \}) + (1-\lambda) C(\Delta_a(\rho), \{ \Pi_a \})\nonumber\\
            &=& \lambda C(\rho, \{ \Pi_a \})\nonumber\\
            &\leq& C(\rho, \{ \Pi_a \}),
        \end{eqnarray}
        where we have used convexity in the first line.
    \end{proof}

\subsubsection{Non-increasing under partial trace}
\label{appsub: prop04}
    We prove the following property under partial trace for our coherence,
    \begin{equation}
        \label{eq: ParTra}
        C(\rho_{12}, \{\Pi_a \otimes I\}) \geq C(\rho_{1}, \{\Pi_a\}),
    \end{equation}
    where $\rho_1 = \text{Tr}_2~ \rho_{12} $ which is the state after partial trace of the product state.
    
    \begin{proof}
        \label{pr: ParTra}
        Consider the optimum decomposition of $\rho_{12}$,
        \begin{equation}
            \rho_{12} = \sum_{k} p_k \ket{\Psi_k}_{12} \bra{\Psi_k},
        \end{equation}
        We have the partial trace of $\rho_{12}$
        \begin{equation}
            \rho_1 = \Tr_2 \rho_{12}.
        \end{equation}
        If we define
        \begin{equation}
            \sigma_1^k \coloneqq \Tr_2 \left( \ket{\Psi_k}_{12} \bra{\Psi_k} \right),
        \end{equation}
        we also have
        \begin{equation}
        \label{eq: rhok}
            \rho_1 = \sum_k p_k \sigma^k_1.
        \end{equation}
        The coherence for one of the pure states in the decomposition can be expressed as
        \begin{eqnarray}
        \label{eq: parTra}
            C\left( \ket{\Psi_k}_{12} \bra{\Psi_k}, \{ \Pi_a \otimes I_2 \} \right) &=& -1 + \sum_a \sqrt{p_a(\Psi)}\nonumber\\
            &=& R(\sigma_1^k, \{ \Pi_a \}),
        \end{eqnarray}
        where we define
        \begin{equation}
            R(\rho, \{ \Pi_a \}) \coloneqq -1 + \sum_a \sqrt{\Tr\left( \rho \Pi_a \right) }
        \end{equation} for the last line and
        \begin{eqnarray}
            p_a(\Psi) &\coloneqq& \Tr \left[ (\Pi_a \otimes I_2) \ket{\Psi_k}_{12} \bra{\Psi_k} \right]\nonumber\\
            &=& \Tr \left[ (\Pi_a \otimes I_2) \ket{\Psi_k}_{12} \bra{\Psi_k} (\Pi_a \otimes I_2) \right]\nonumber\\
            &=& \sum_{bc} ({}_1\bra{b} \Pi_a \otimes   {}_2\bra{c}) \ket{\Psi_k}_{12} \bra{\Psi_k} (\ket{b}_1 \Pi_a \otimes   \ket{c}_2)\nonumber\\
            &=& \sum_b {}_1\bra{b}\Pi_a \Tr_2 \left( \ket{\Psi_k}_{12}\bra{\Psi_k} \right) \Pi_a \ket{b}\nonumber\\
            &=& \sum_b {}_1\bra{b}\Pi_a \sigma_1^k\Pi_a \ket{b}\nonumber\\
            &=& \Tr_1 \left( \sigma^k_1 \Pi_a \right).
        \end{eqnarray}
        From the relation in eq. \eqref{eq: EnsConv05}, we can get for mixed states,
        \begin{eqnarray}
        \label{eq: conApp}
            C(\sigma_1^k, \Pi_a) &\leq& R(\sigma_1^k, \{ \Pi_a \}).
        \end{eqnarray}
        Finally, from the convex-roof construction nature and eq. \eqref{eq: rhok}, we have
        \begin{eqnarray}
        \label{eq: conroofApp}
            C(\rho_1, \{\Pi_a\}) \leq \sum_k C(\sigma_1^k, \{\Pi_a\}).
        \end{eqnarray}
        Using the results from eqs. \eqref{eq: ParTra}, \eqref{eq: conApp} and \eqref{eq: conroofApp}, we have
        \begin{eqnarray}
            C(\rho_{12}, \{ \Pi_a \otimes I_2 \}) &=& \sum_k p_k C(\ket{\Psi_k}_{12} \bra{\Psi_k}, \{ \Pi_a \otimes I_2 \})\nonumber\\
            &=& \sum_k p_k R(\sigma_1^k, \{ \Pi_a \})\nonumber\\
            &\geq& \sum_k p_k C(\sigma_1^k, \{ \Pi_a \})\nonumber\\
            &\geq& C(\rho_1, \{ \Pi_a \}).
        \end{eqnarray}
    \end{proof}

\subsubsection{Unitarily covariant}
\label{appsub: prop05}

    \begin{proof}
        \label{pr: UniCov}
        In order to prove unitary covariance, first, consider the effect of the unitary operation on the dephased state,
        \begin{eqnarray}
            \label{der: UniCovA}
            \Delta(U\ket{\psi_i} \| \{U\Pi_a U^\dagger\}) &=& \sum_a U\Pi_a U^\dagger U\ket{\psi_i}\bra{\psi_i}U^\dagger U\Pi_a U^\dagger\nonumber\\
            &=& U \sum_a \Pi_a \ket{\psi_i}\bra{\psi_i}\Pi_a U^\dagger\nonumber\\
            &=& U \Delta(\ket{\psi_i} \| \{\Pi_a \}) U^\dagger\nonumber\\
            &=& U \Delta_a(\ket{\psi_i}) U^\dagger.
        \end{eqnarray}
        The resulting dephased state also transforms as unitarily. Substituting this to the RHS of the definition of the coherence, we get
        \begin{eqnarray}
            \label{der: UniCovB}
            \sum_i p_i C(U\ket{\psi_i}\bra{\psi_i}U^\dagger, \{ U\Pi_a U^\dagger \})
            &=& \sum_i p_i C \left( 1- \frac{1}{D} D_B^2(\Delta(U\ket{\psi_i} \| \{U\Pi_a U^\dagger\}) \| \Delta_a(\ket{\psi_{mc}}) \right)\nonumber\\
            &=& \sum_i p_i C \left( 1- \frac{1}{D} D_B^2 \left(U\Delta_a(\ket{\psi_i})U^\dagger \| U\Delta_a(\ket{\psi_{mc}})U^\dagger \right) \right)\nonumber\\
            &=& \sum_i p_i C \left( 1- \frac{1}{D} D_B^2(\Delta_a(\ket{\psi_i}) \| \Delta_a(\ket{\psi_{mc}}) \right)\nonumber\\
            &=& \sum_i p_i C(\ket{\psi_i}\bra{\psi_i}, \{ \Pi_a \}),
        \end{eqnarray}
        where we have used the unitary invariant property of Fidelity in the third line.
        Using the above result, we can confirm the property,
        \begin{eqnarray}
            \label{der: UniCovC}
            C(U\rho U^\dagger, \{ U\Pi_a U^\dagger \})
            &=& \inf_{\{p_i, \ket{\psi_i}\}} \sum_i p_i C(U\ket{\psi_i}\bra{\psi_i}U^\dagger, \{ U\Pi_a U^\dagger \})\nonumber\\
            &=& \inf_{\{p_i, \ket{\psi_i}\}} \sum_i p_i C(\ket{\psi_i}\bra{\psi_i}, \{ \Pi_a \})\nonumber\\
            &=& C(\rho, \{ \Pi_a \}).
        \end{eqnarray}
    \end{proof}

\subsubsection{Invariant under unitary transformation that commutes with a Hermitian operator $A$ that has $\{\Pi_a\}$ as the eigenvectors (i.e. $A = \sum_a a \ket{a}\bra{a}$ with $a \in \mathbb{R}$).}
\label{appsub: prop06}

    \begin{proof}
        \label{pr: UniInv}
        Since $U$ commutes with $A$, we can write the unitary operation as
        \begin{equation}
            U\ket{a} = e^{i\theta_a} \ket{a},
        \end{equation}
        where $\theta_a \in \mathbb{R}$. Thus, when applied $\Pi_a$ to the unitary operation $U$, we have
        \begin{equation}
            \Pi_a U = e^{i\theta_a} \Pi_a.
        \end{equation}
        The effect of this to the dephased state is
        \begin{eqnarray}
            \Delta(U\ket{\psi_i} \| \{\Pi_a\}) &=& \sum_a \Pi_a  U\ket{\psi_i}\bra{\psi_i}U^\dagger \Pi_a \nonumber\\
            &=& \sum_a e^{i\theta_a} \Pi_a \ket{\psi_i}\bra{\psi_i}\Pi_a e^{-i\theta_a}\nonumber\\
            &=& \Delta_a(\ket{\psi_i}).
        \end{eqnarray}
        We can plug this into the RHS of the definition of the coherence in any decomposition,
        \begin{eqnarray}
            \label{der: UniInvA}
            \sum_i p_i C(U\ket{\psi_i}\bra{\psi_i}U^\dagger, \{ \Pi_a \})
            &=& \sum_i p_i C\left( 1- \frac{1}{D} D_B^2(\Delta_a(U\ket{\psi_i}) \| \Delta_a(\ket{\psi_{mc}}) \right)\nonumber\\
            &=& \sum_i p_i C\left( 1- \frac{1}{D} D_B^2(\Delta_a(\ket{\psi_i}) \| \Delta_a(\ket{\psi_{mc}}) \right)\nonumber\\
            &=& \sum_i p_i C(\ket{\psi_i}\bra{\psi_i}, \{ \Pi_a \}).
        \end{eqnarray}
        Using the above result, we can confirm this invariance,
        \begin{eqnarray}
            \label{der: UniInvB}
            C(U\rho U^\dagger, \{ \Pi_a \}) &=& \inf_{\{p_i, \ket{\psi_i}\}} \sum_i p_i C(U\ket{\psi_i}\bra{\psi_i}U^\dagger, \{ \Pi_a \})\nonumber\\
            &=& \inf_{\{p_i, \ket{\psi_i}\}} \sum_i p_i C(\ket{\psi_i}\bra{\psi_i}, \{ \Pi_a \})\nonumber\\
            &=& C(\rho, \{ \Pi_a \}).
        \end{eqnarray}
    \end{proof}

\subsubsection{Invariant under permutation of the incoherent basis}
\label{appsub: prop07}
\begin{proof}
    Consider an operation that permutes the incoherent basis up to a phase,
    \begin{equation}
        \phi(\ket{a}) = e^{i\theta_a} \ket{\mu(a)},
    \end{equation}
    where $\theta_a \in \mathbb{R}$ and $\mu(a)$ does a systematic permutation to the incoherent basis.

    Next, we consider the fidelity of an arbitrary pure state in this basis. First, the pure state can generally be written in the incoherent basis,
    \begin{equation}
        \ket{\psi_i}\bra{\psi_i} = \sum_{ab} p_{ab} \ket{a}\bra{b},
    \end{equation}
    with $p_{ab} \in \mathbb{R}$ and satisfies the constraints of being the entries of a density matrix. The dephased state, thus, can be written as
    \begin{eqnarray}
        \label{der: InvPer1}
        \Delta_a (\ket{\psi_i}\bra{\psi_i}) &=& \sum_{a} \Pi_a \ket{\psi_i}\bra{\psi_i} \Pi_a\nonumber\\
        &=& \sum_{abc} p_{bc} \ket{a}\bra{a} \ket{b}\bra{c} \ket{a}\bra{a}\nonumber\\
        &=& \sum_{a} p_{aa} \ket{a}\bra{a}.
    \end{eqnarray}
    Plugging this into the fidelity, we have
    \begin{eqnarray}
        \label{der: InvPer2}
        \sqrt{F}(\Delta_a(\ket{\psi_i}\bra{\psi_i}) \| I/d) &=& \frac{1}{\sqrt{d}}\text{tr} \sqrt{\Delta_a(\ket{\psi_i}\bra{\psi_i})}\nonumber\\
        &=& \frac{1}{\sqrt{d}}\sum_{a} \sqrt{p_{aa}}.
    \end{eqnarray}

    Now, consider permuting the pure state using the permutation operation $\phi$,
    \begin{eqnarray}
        \label{der: InvPer3}
        \phi(\ket{\psi_i}\bra{\psi_i}) &=& \sum_{ab} p_{ab}~ \phi(\ket{a}\bra{b})\nonumber\\
        &=& \sum_{ab} p_{ab}~ e^{i(\theta_a-\theta_b)} \ket{\mu(a)}\bra{\mu(b)}.
    \end{eqnarray}
    We can insert this to the dephasing operation in order to find the resulting dephased state,
    \begin{eqnarray}
        \label{der: InvPer4}
        \Delta_a(\phi(\ket{\psi_i}\bra{\psi_i})) &=& \sum_{abc} p_{bc}~ e^{i\theta_b-i\theta_c} \braket{a}{\mu(b)}\braket{\mu(c)}{a} \ket{a}\bra{a}\nonumber\\
        &=& \sum_{a} p_{\mu^{-1}(a),\mu^{-1}(a)} \ket{a}\bra{a}.
    \end{eqnarray}
    The fidelity for this state becomes
    \begin{eqnarray}
        \label{der: InvPer5}
        \sqrt{F}(\Delta_a(\phi(\ket{\psi_i}\bra{\psi_i})) \| I/d) &=&  \text{tr} \sqrt{\Delta_a(\phi(\ket{\psi_i}\bra{\psi_i}))}\nonumber\\
        &=& \frac{1}{\sqrt{d}} \sum_{a} \sqrt{ p_{\mu^{-1}(a),\mu^{-1}(a)}}\nonumber\\
        &=& \frac{1}{\sqrt{d}} \sum_{a} \sqrt{p_{aa}}\nonumber\\        
        &=& \sqrt{F}(\Delta_a(\ket{\psi_i}\bra{\psi_i}) \| I/d),\nonumber\\
    \end{eqnarray}
    where we have used the fact that permuting the order of the summation doesn't change the result of the summation itself in the third line. We have 
    \begin{eqnarray}
        D_B^2(\Delta_a(\phi(\ket{\psi_i}\bra{\psi_i})) \| \Delta_{a}(\ket{\psi_{mc}}))
        &=& D_B^2(\Delta_a(\ket{\psi_i}\bra{\psi_i}) \| \Delta_{a}(\ket{\psi_{mc}})),
    \end{eqnarray}
    which results in a preserved coherence for pure states.
    \begin{equation}
        C(\phi(\ket{i}\bra{i}), \{\Pi_a\}) = C(\ket{i}\bra{i}, \{\Pi_a\}).
    \end{equation}

    Lastly, consider if we have a mixed quantum density matrix $\rho$ with an optimum decomposition $\rho = \sum_{i} p_i \ket{i}\bra{i}$. We can write the coherence to be
    \begin{equation}
         C(\rho, \{\Pi_a\}) = \sum_{i} p_i C(\ket{i}\bra{i}, \{\Pi_a\}).
    \end{equation}
    Since the permuting operation doesn't change the value of the coherence for any decomposition, the optimal decomposition doesn't change, and we can still use it for the coherence,
    \begin{eqnarray}
        \label{der: InvPer6}
        C(\phi(\rho), \{\Pi_a\}) &=& \sum_{i} p_i C(\phi(\ket{i}\bra{i}), \{\Pi_a\})\nonumber\\
        &=& \sum_{i} p_i C(\ket{i}\bra{i}, \{\Pi_a\})\nonumber\\
        &=& C(\rho, \{\Pi_a\}).
    \end{eqnarray}
    \end{proof}

\subsection{Proof for the wave-particle duality relation for mixed states}
\label{app: Wpd}
\begin{proof}
    First, we begin by evaluating the relation between the Bures distance of the dephased arbitrary state and the maximally mixed state with its form in the convex–roof construction decompositions, where it is a function of a decomposition of pure states of the arbitrary state,
    \begin{eqnarray}
        \label{der: ConvRof}
        &\sum_{i}& p_i D_B^2 (\Delta_a(\ket{\psi_i}) \| \Delta_a(\ket{\psi_{mc}}))\nonumber\\
        &=& \sum_{i} p_i ~2 \left(1-\sqrt{F} (\Delta_a(\ket{\psi_i}) \| \Delta_a(\ket{\psi_{mc}})) \right)\nonumber\\
        &=& \sum_{i} p_i ~2 \left(1- \frac{1}{\sqrt{d}} \text{Tr}\sqrt{\Delta_a(\ket{\psi_i})} \right)\nonumber\\
        &=& 2 \left(1- \frac{1}{\sqrt{d}} \sum_{a} \sum_{i} p_i \sqrt{\text{Tr}\left(\Pi_a\ket{\psi_i}\bra{\psi_i}\Pi_a\right)}\right)\nonumber\\
        &\geq& 2 \left(1- \frac{1}{\sqrt{d}} \sum_{a} \sqrt{ \sum_{i} p_i \text{Tr} (\Pi_a\ket{\psi_i}\bra{\psi_i}\Pi_a)} \right) \nonumber\\
        &=& 2 \left(1- \frac{1}{\sqrt{d}} \sum_{a} \sqrt{ \text{Tr} (\Pi_a\rho\Pi_a)} \right) \nonumber\\
        &=& 2 \left(1- \frac{1}{\sqrt{d}} \text{Tr}\sqrt{\Delta_a(\rho)} \right)\nonumber\\
        &=& 2 \left(1-\sqrt{F} (\Delta_a(\rho) \| \Delta_a(\ket{\psi_{mc}})) \right)\nonumber\\
        &=& D_B^2 (\Delta_a(\rho) \| \Delta_a(\ket{\psi_{mc}})),
    \end{eqnarray}
    where we have used Jensen's inequality in the fourth line. Since the Bures distance for any decomposition of the arbitrary quantum state is always larger than or equal to the Bures distance of the original state, the supremum of the decomposition is also larger than or equal to the version of the original state. This result is of great importance to us because our coherence is the infimum of the complement of the normalised expression for the decomposition of the Bures distance. With a little bit of thought, we can infer the following relation,
    \begin{eqnarray}
        \label{eq: ConvRofCoh}
        &&\frac{C(\rho, \{\Pi_a\})}{\sqrt{d}-1}\nonumber\\
        &=& \inf_{\{ p_i, \ket{\psi_i} \}} 1 - \frac{\sum_{i} p_i D_B^2 (\Delta_a(\ket{\psi_i}) \| \Delta_a(\ket{\psi_{mc}}))}{2\left( 1-\frac{1}{\sqrt{d}} \right)}\nonumber\\
        &\leq& 1 - \frac{D_B^2 (\Delta_a(\rho) \| \Delta_a(\rho_{mc}))}{2\left( 1-\frac{1}{\sqrt{d}} \right)},
    \end{eqnarray}
    which leads to the wave–particle duality relation expressed above,
    \begin{equation}
        \label{eq: ConvRofCohEl}
        \frac{C(\rho, \{\Pi_a\})}{\sqrt{d}-1} + \frac{D_B^2 (\Delta_a(\rho) \| \Delta_a(\rho_{mc}))}{2\left( 1-\frac{1}{\sqrt{d}} \right)} \leq 1.
    \end{equation}
\end{proof}

\subsection{Proof for the relation between $H_{min}$ with Predictability $P$ and Coherence $C$}
\label{app: Unp}
\begin{proof}

First, rewrite the quantities in the following way,
\begin{equation}
    \widetilde{C}(\rho, \{\Pi_a\}) \coloneqq 1 - \bar{P}(\rho, \{\Pi_a\}),
\end{equation}
where $\widetilde{C}(\rho, \{\Pi_a\}) = \frac{C(\rho, \{\Pi_a\})}{\max_{\rho} C(\rho, \{\Pi_a\})}$ is the normalised coherence and
\begin{equation}
    \bar{P}(\rho, \{\Pi_a\}) \coloneqq \sup_{\{p_i, \ket{\psi_i}\}} \sum_{i} p_i \widetilde{P}(\ket{\psi_i}, \{\Pi_a\}),
\end{equation}
is the supremum of the averaged normalised predictability $\widetilde{P}(\ket{\psi_i}, \{\Pi_a\}) = \frac{P(\ket{\psi_i}, \{\Pi_a\})}{\max_{\ket{\psi_i}}P(\ket{\psi_i}, \{\Pi_a\})}$ with respect to the decomposition $\rho = \sum_i p_i \ket{\psi_i}$. On the other hand, it's simpler to write the normalised predictability $\widetilde{P}(\rho, \{\Pi_a\})$ as
\begin{equation}
    \widetilde{P}(\rho, \{\Pi_a\}) \coloneqq \frac{\sqrt{d} - s}{\sqrt{d} - 1},
\end{equation}
where
\begin{equation}
    s \coloneqq \sum_a \sqrt{\Tr{ \left(\Pi_a \rho\right)}} = \sum_a \sqrt{p_a}.
\end{equation}
Now, we introduce a new quantity $\bar{s}$ such that
\begin{equation}
    \bar{P}(\rho, \{\Pi_a\}) \coloneqq \frac{\sqrt{d} - \bar{s}}{\sqrt{d} - 1}.
\end{equation}

Using eq. \eqref{eq: WavParBurCoh} , we can find the relation between $s$ and $\bar{s}$, which is
\begin{eqnarray}
    \widetilde{C}(\rho, \{\Pi_a\}) &\leq& 1 - \widetilde{P}(\rho, \{\Pi_a\}),\nonumber\\
    1 - \bar{P}(\rho, \{\Pi_a\}) &\leq& 1 - \widetilde{P}(\rho, \{\Pi_a\}),\nonumber\\
    \bar{P}(\rho, \{\Pi_a\}) &\geq& \widetilde{P}(\rho, \{\Pi_a\}),\nonumber\\
    \label{eq: srel}
    \bar{s} &\leq& s.
\end{eqnarray}

Next, we find the relation between $s$ and the maximum probability $p_{max} \in \{p_a\}$. $p_{max}$ is the optimal guessing probability for $\{\Pi_a\}$, $\max_{a} \text{Pr}(A)_\rho$. We start by decomposing $s$ in terms of the maximum probability and the other probabilities,
\begin{equation}
    \label{eq: s}
    s = \sqrt{p_{max}} + \sum_{a \neq max} \sqrt{p_a}.
\end{equation}
Using Cauchy-Schwarz inequality for the sum of the probabilities in the second term, we have the following relation,
\begin{equation}
    \left(\sum_{a \neq max} \sqrt{p_a}\right)^2 \leq (d-1) \left( \sum_{a \neq max} p_a \right).
\end{equation}
Substituting eq. \eqref{eq: s} into the inequality, we will have the following calculation.

\begin{eqnarray}
\label{eq: pmaxQd}
    \left( s - \sqrt{p_{max}} \right)^2 &\leq& (d-1) (1-p_{max}),\nonumber\\
    s^2 - 2 s \sqrt{p_{max}} + p_{max} &\leq& d - 1 - d ~p_{max} + p_{max},\nonumber\\
    s^2 - 2 s \sqrt{p_{max}} &\leq& d - 1 - d ~p_{max},\nonumber\\
    p_{max} - 2 \frac{s}{d} \sqrt{p_{max}} &\leq& \frac{d - 1 -s^2}{d},\nonumber\\
    \left( \sqrt{p_{max}} - \frac{s}{d} \right)^2 - \frac{s^2}{d^2} &\leq& \frac{d - 1 -s^2}{d},\nonumber\\
    \left( \sqrt{p_{max}} - \frac{s}{d} \right)^2 &\leq& \frac{s^2 + d^2 - d - ds^2}{d^2}\nonumber\\
    &=& \frac{(d-s^2)(d-1)}{d^2},\nonumber\\
    \sqrt{p_{max}} &\leq& \frac{ s + \sqrt{(d-s^2)(d-1)}}{d},\nonumber\\
    p_{max} &\leq& \left[ \frac{ s + \sqrt{(d-s^2)(d-1)}}{d} \right]^2,\nonumber\\
\end{eqnarray}
where we have used the fact that $1/d \leq s/d \leq 1/\sqrt{d} \leq \sqrt{p_{max}} \leq 1$ for the eighth line.

Now define a function for the LHS of the equation $Q_d(s)$,
\begin{equation}
    Q_d(s) \coloneqq \left[ \frac{ s + \sqrt{(d-s^2)(d-1)}}{d} \right]^2.
\end{equation}
By finding the derivative of the function $Q_d(s)$ with respect to $s$,
\begin{eqnarray}
    \frac{dQ_d(s)}{ds} &=& 2 \frac{ s + \sqrt{(d-s^2)(d-1)}}{d^2} \left( 1 - s \sqrt{\frac{d-1}{d-s^2}} \right)\nonumber\\
    &\leq& 0,
\end{eqnarray}
we know that this function is monotonically non-increasing in $s \in [1, \sqrt{d}]$. This means that because of the eq. \eqref{eq: srel}, the following relation holds,
\begin{equation}
    Q_d(s) \leq Q_d(\bar{s}).
\end{equation}
The second derivative of the function,
\begin{eqnarray}
    \frac{d^2Q_d(s)}{ds^2} &=& - \frac{2}{d^2  (d-s^2)^{3/2}} \Big[ \sqrt{d-1} ~s^3 + (d-2) (d-s^2)^{3/2}\nonumber\\
    &&~~ + 3 \sqrt{d-1} ~(d-s^2) ~s  \Big]\nonumber\\
    &\leq& 0,
\end{eqnarray}
shows that the function $Q_d(s)$ is concave.

Consider the term inside the summation of the label-assisted guessing probability $\text{Pr}^{cl}_{guess} (A|I)_\rho$ in eq.~\eqref{eq: clOptPr},
\begin{eqnarray}
    \max_a \bra{\psi_i} \Pi_a \ket{\psi_i} &=& \max_a \Tr \left( \ket{\psi_i} \bra{\psi_i} \Pi_a \right)\nonumber\\
    &=& \max_a p_a(\ket{\psi_i})\nonumber\\
    &=& p_{max}(\ket{\psi_i}).
\end{eqnarray}
This is just the same maximum optimal guessing probability as introduced previously,

\begin{eqnarray}
    \sum_i p_i ~p_{max}(\ket{\psi_i}) &\leq& \sum_i p_i ~Q_d(s(\ket{\psi_i}))\nonumber\\
    &\leq& Q_d\left( \sum_i p_i s(\ket{\psi_i}) \right)\nonumber\\
    &\leq& Q_d(\bar s(\rho)),
\end{eqnarray}
where we have used eq.~\eqref{eq: pmaxQd} in the first line, the concavity of $Q_d(s)$ in the second line and the fact that $\sum_i p_i s(\ket{\psi_i}) \geq \bar s(\rho)$ in last line. Since the previous result holds for any decomposition of $\rho$, then
\begin{eqnarray}
    \text{Pr}^{cl}_{guess}(A|I)_\rho &=& \sup_{\{ p_i, \ket{\psi_i}\}} \sum_i p_i ~p_{max}(\ket{\psi_i})\nonumber\\
    &\leq& Q_d(\bar s).
\end{eqnarray}
Therefore, the classical-label conditional min-entropy  has the following bound,
\begin{eqnarray}
    H^{cl}_{min}(A|I)_\rho\nonumber
    &=& - \log \text{Pr}^{cl}_{guess}(A|I)_\rho\nonumber\\
    &\geq& -\log{Q_d(\bar{s})}\nonumber\\
    &=& -\log{Q_d\left(1+C(\rho, \{\Pi_a\})\right)}.\nonumber\\
\end{eqnarray}

\end{proof}